\documentclass[letterpaper, 10 pt, conference]{ieeeconf}  % Comment this line out if you need a4paper

\IEEEoverridecommandlockouts                              % This command is only needed if 
\usepackage{graphicx} % for pdf, bitmapped graphics files
\usepackage{amsmath} % assumes amsmath package installed
\usepackage{amssymb}  % assumes amsmath package installed
\usepackage{subcaption}
\usepackage{psfrag}
\usepackage{xcolor}%

\newtheorem{definition}{\bfseries Def.}
\newtheorem{theorem}{\bfseries Theorem}
\newtheorem{lemma}{\bfseries Lemma}

\let\emptyset\varnothing

\newtheorem{assumption}{\bf Assumption}

\title{\LARGE \bf
The Small-Talk Effect in Practical Synchronization\\ of Heterogeneous Oscillating Dynamics
}

\author{Finn Voland$^{1}$, Vincent Schmidtke$^{1}$, Zonglin Liu$^{1}$ and Olaf Stursberg$^{1}$% <-this % stops a space
\thanks{$^{1}$The authors are with the Institute of Control and System Theory, EECS, 
University of Kassel, Germany 
        {\tt\small $\{$uk089421, v.schmidtke, Z.Liu, stursberg$\}$@uni-kassel.de}}%
}

\begin{document}

\maketitle
%\thispagestyle{empty}
%\pagestyle{plain}

%%%%%%%%%%%%%%%%%%%%%%%%%%%%%%%%%%%%%%%%%%%%%%%%%%%%%%%%%%%%%%%%%%%%%%%%%%%%%%%%
\begin{abstract}

This paper investigates communication schemes in synchronization of  heterogeneous Liénard oscillator systems with time-triggered coupling. Existing results rely on global exchange of information at discrete time instants, providing limited flexibility in adjusting  the trade-off between communication rate and synchronization performance, while not directly extending to settings of multiple time scales. A control scheme is proposed in which global updates are complemented by more frequent local interactions within a sub-network. For a prescribed bound on the synchronization error, conditions are derived under which the global and local update rates can be selected as design parameters. This allows to ensure that trajectories remain in a neighborhood of the synchronized periodic orbit while satisfying the error bound always. A numerical example illustrates the beneficial effect of the local interactions.

\end{abstract}

%%%%%%%%%%%%%%%%%%%%%%%%%%%%%%%%%%%%%%%%%%%%%%%%%%%%%%%%%%%%%%%%%%%%%%%%%%%%%%%%
\section{INTRODUCTION}

The emergence of synchronized behavior in networks of dynamical systems is a fundamental phenomenon in biological and engineered systems. Examples include locomotion, circadian rhythms, and cardiac dynamics, which are commonly modeled using nonlinear coupled oscillators.
Among the existing classes of oscillator models, Liénard systems have proven particularly suitable for capturing such oscillatory and synchronous behaviors \cite{dutra2003modeling, rompala2007dynamics, kongas1999bifurcation, dos2004rhythm}. For sets of  Liénard systems, as well as for similar classes of oscillators, several investigations have shown that practical synchronization can be achieved by using mechanism of diffusive coupling, including the work in  \cite{hill2008global, montenbruck2015practical, kim2015robustness, 7809144, lee2020tool, lee2019behavior}. It is shown in  \cite{coraggio2020distributed} that stronger notion of synchronization, namely local asymptotic synchronization, can be obtain by introducing discontinuous sign-based coupling.
All these results refer to settings in which coupling signals are continuously present -- however, the realization of coupling requires significant effort (as use of energy or material for establishing the interaction of communication), the question arises to which extent can synchronization be achieved with only infrequent or periodic coupling. Along this line, \cite{tanwani2021lyapunov} proposed a time-triggered control scheme in which information is exchanged only at discrete time instants. By use of singular perturbation, the work shows that local practical synchronization is achieved if puls-like coupling is applied with sufficient frequency. The idea is extended in \cite{tanwani2024singularly} to allow for sequences of interaction graphs, ensuring synchronization under arbitrarily switching between interaction graphs if the union of all graphs is connected.

However, these approaches rely on the separation of time scales for the coupling and the system evolution, which does not directly extend to settings in which multiple intrinsic time scales are present in the oscillator dynamics. In addition, no quantitative insight is provided into how the structure of the coupling topology and the timing of discrete information exchange jointly affect the synchronization performance.

To address this gap, this paper introduces a time-triggered control scheme in which global information exchange is interleaved with additional local interactions among a subset of oscillators -- the latter is referred to as \textit{small-talk}. The considered concept of having multiple communication frequencies is inspired by literature on opinion dynamics \cite{Kravitzch2023, Bertotti2024}, in which individual-to-individual interactions are modeled by less frequent updates, while more frequent updates capture opinion formation driven by online communication. The potential positive influence of small-talk on information propagation is supported by \cite{Tump2024}, which demonstrates that earlier dissemination of information through a network yields greater downstream influence. Building on this, the present work investigates how the timing and distribution of small-talks affect and improve synchronization accuracy.

The paper is structured into Sec. 2 to introduce the setting and problem, while Sec. 3 presents the main result on practical synchronization, building on infrequent information exchange of all oscillators combined with additional intermediate updates for a subset of oscillators.  Section 4 provides a numerical example illustrating the effect of the proposed update scheme, and Sec. 5 concludes the paper.

\section{HYBRID NETWORK DYNAMICS WITH TWO-TIME-SCALE NETWORK UPDATES}

Consider a set of $N$ heterogeneous Liénard oscillators, where each
oscillator $i \in \mathcal{N} = \{1, \dots, N\}$ has a two-dimensional state
$(\omega_i(t), \xi_i(t))$ governed by the continuous-time dynamics:
\begin{equation}\label{eq:dyn}
\begin{aligned}
\dot{\omega}_i(t) &= -\omega_i(t) + \xi_i(t), \\
\dot{\xi}_i(t) &= \underbrace{(1 - f_i(\omega_i(t)))\bigl(-\omega_i(t) +
\xi_i(t)\bigr) - g_i(\omega_i(t))}_{=:\,\phi_i(\omega_i(t),\,\xi_i(t))},
\end{aligned}
\end{equation}
where the nonlinear term $\phi_i(\omega_i(t), \xi_i(t))$ is composed of the
damping function $f_i(\omega_i(t))$ and the restoring force function
$g_i(\omega_i(t))$. Starting from the initial time $t_0$,   the set of oscillators is assumed to exchange state  information at discrete times:
\begin{equation}\label{eq:Tg}
t_{g,k} \in\mathcal{T}_g = \left\{t_0 + k \Delta T,\;
k \in \mathbb{N}_0
\right\},
\end{equation}
with update interval $\Delta T > 0$.
The local information is exchanged over  an undirected and connected graph $\mathcal{G} = (\mathcal{V},\mathcal{E})$ with the set of vertices $\mathcal{V}=\mathcal{N}$ the set $\mathcal{E}\subseteq \mathcal{V}\times\mathcal{V}$ of edges. In the  adjacency matrix $A=[a_{ij}]$ of  $\mathcal{G}$ , $a_{ij}=1$ denotes that an edge exists between the oscillators $i$ and $j$ (thus exchange of information is possible), and $a_{ij}=0$ otherwise. Let $L$ denote the  Laplacian matrix of $\mathcal{G}$. It is known that $L$ is symmetric and positive semi-definite with a simple eigenvalue of zero, i.e., $\lambda_1(L)=0< \lambda_2(L) \leq \lambda_3(L) \leq \ldots \leq \lambda_N(L)$.
Following \cite{tanwani2021lyapunov}, the  information exchange between the oscillators at time $t_{g,k} \in \mathcal{T}_g$ is described  by:
\begin{equation}\label{eq:globup}
\omega_i(t^+_{g,k}) := \omega_i(t^-_{g,k}), 
\qquad
\xi_i(t^+_{g,k}) := \sum_{j=1}^N w_{ij}\,\xi_j(t^-_{g,k})
\end{equation}
and the constants $w_{ij}$ are the entries of the matrix:
\begin{equation}\label{eq:W}
W := I_N - \varphi L,~~\varphi>0.
\end{equation}
The positive constant $\varphi$ (also referred to as the coupling gain) is assumed to be  chosen to make $W$  a non-negative matrix.
The main objective in \cite{tanwani2021lyapunov, tanwani2024singularly} is to investigate whether a set of heterogeneous oscillators can be synchronized through the information exchange described by \eqref{eq:globup}. The authors model the oscillator dynamics by a hybrid system, combining the continuous dynamics in \eqref{eq:dyn} with the discrete updates in \eqref{eq:globup}. If the updates occur infinitely often within a finite time horizon ($\Delta T \to 0$), the oscillators are guaranteed to asymptotically synchronize to the limit cycle $\mathcal{A}_0 \subset \mathbb{R}^2$ (if this exists) of the following second-order dynamics (also referred to as the \emph{averaged dynamics}):
\begin{align} \label{eq:avgdynamics} 
    \begin{bmatrix}
        \dot{w} \\ \dot{\xi}
    \end{bmatrix}\hspace{-0.8mm} =\hspace{-0.8mm} \begin{bmatrix}
        -\omega +\xi \\
   (1 - \frac{1}{N}\sum\limits_{i=1}^N f_i(w) (-w +\xi) - \frac{1}{N}\sum\limits_{i=1}^N g_i(w)  )
    \end{bmatrix}.
\end{align}
By treating $\Delta T$ as a perturbation parameter, the method of singular perturbation analysis (see \cite{khalil2002nonlinear}) is adopted in \cite{tanwani2021lyapunov}. For $\Delta T$ smaller than an upper bound  $\Delta T_{max}$, it is shown that the oscillators with heterogenous dynamics can be practically synchronized to a bounded neighborhood of $\mathcal{A}_0$, i.e., it holds with an $\eta>0$ and the distance:
\begin{align}  
| (w_i(t), \xi_i(t) ) |_{\mathcal{A}_0}:= \inf_{x \in \mathcal{A}_0}||x -  \begin{bmatrix}   w_i(t) \\ \xi_i(t) \end{bmatrix}  ||  \label{eq:boundA03}
\end{align}
between $ \begin{bmatrix}   w_i(t) \\ \xi_i(t) \end{bmatrix}$ and the  limit cycle $\mathcal{A}_0$ that:
\begin{align}  
&\underset{t \to \infty}{\text{lim sup}} |    (w_i(t), \xi_i(t) ) |_{\mathcal{A}_0} \le \eta,~~\forall~ i \in \{1,\dots,N\}  \label{eq:boundA02}.
\end{align}
\begin{figure}[t]
    \centering
    \begin{subfigure}{0.56\columnwidth}
        \centering
        \psfrag{G}[][]{$\mathcal{G}$}
        \psfrag{G1}[][t]{\textcolor{red}{$\mathcal{G}_1$}}
        \includegraphics[width=\linewidth]{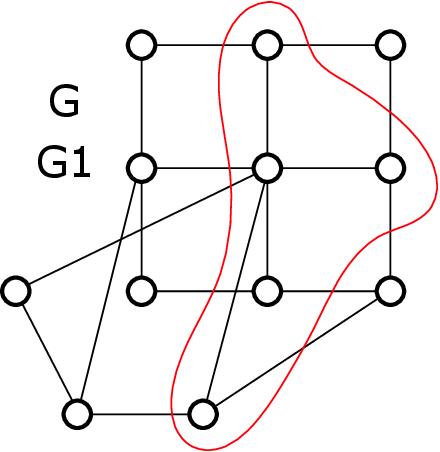}
        \caption{Network of 12 oscillators. Small-talk updates care possible in subgraph $\mathcal{G}_1$.}
        \label{fig:graph}
    \end{subfigure}
    \hfill
    \begin{subfigure}{0.4\columnwidth}
        \centering
        \psfrag{g}[][b]{$\Delta T$}
        \psfrag{l}[][]{$\tau$}
        \psfrag{t}[][]{$t$}
        \psfrag{t0}[][]{$t_0$}
        \includegraphics[width=\linewidth]{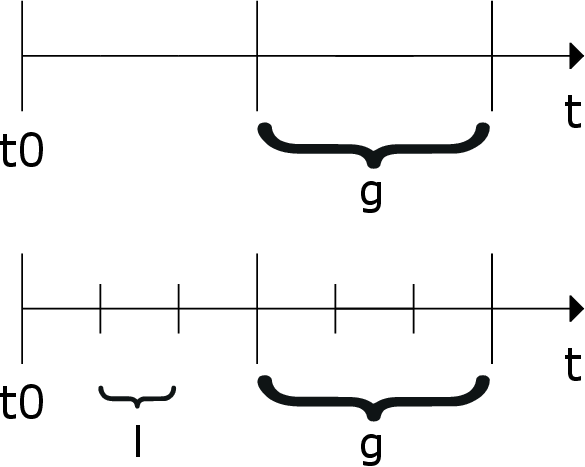}
        \caption{Top: Only regular updates every $\Delta T$. Bottom: Small-talk updates in $\mathcal{G}_1$ every $\tau$, with regular updates after $\mu=3$ small-talk updates.}
        \label{fig:updates}
    \end{subfigure}
    \caption{Time-triggered update scheme with regular updates in $G$ and intermediate small-talk updates in $G_1$. }
    \label{fig:combined}
\end{figure}
In addition to the regular updates defined by \eqref{eq:globup} and used in \cite{tanwani2021lyapunov}, this paper considers the existence of a subset of oscillators that can exchange information along existing edges in $\mathcal{E}$ at a frequency higher than the regular updates. This update scheme is referred to as \emph{small-talk updates}, and is illustrated in Fig.~\ref{fig:graph}. Such an additional update mechanism is can be observed in biological \cite{NUNEZ201675} or social systems \cite{Kravitzch2023}, in which certain subgroups may exchange signals with higher frequency in case of close proximity.

Formally, let the set $\mathcal{N}$ be decomposed into two subsets $ \mathcal{N}_1$ and $\mathcal{N}_2$ such that $\mathcal{N} = \mathcal{N}_1 \cup \mathcal{N}_2$ and $ \mathcal{N}_1 \cap \mathcal{N}_2 = \emptyset$. The subset $\mathcal{N}_1$  is assumed to contain   $N_1$ oscillators capable of generating small-talk updates, while  $\mathcal{N}_2$ contains the remaining oscillators. Assuming that small-talk updates must use existing edges in $\mathcal{E}$, the Laplacian matrix $L$ of the original graph  $\mathcal{G}$ can be decomposed as follows:
\begin{equation}\label{eq:laplaciandecompse}
L =
\begin{bmatrix}
L_{11} & L_{12} \\
L_{21} & L_{22}
\end{bmatrix},
\end{equation}
in which the submatrix $L_{11}\in \mathbb{R}^{N_1 \times N_1}$ corresponds to the oscillators in $\mathcal{N}_1$, and $L_{22}\in \mathbb{R}^{N_2 \times N_2}$ corresponds to  those in $\mathcal{N}_2$. The set $\mathcal{T}_l $ of discrete time $t_{l}$, at which small talk takes place, is defined by:
\begin{equation}\label{eq:smalltalktime}
\mathcal{T}_l =
\left\{
t_{g,k} + m\tau \;:\;
k \in \mathbb{N}_0,\;
m = 1,\dots,\mu-1
\right\}
\end{equation}
with $\tau>0$ denoting the interval between two small talks. For $\mu \in \mathbb{Z}^+$ applies  $\mu \tau = \Delta T$, i.e., a total of $\mu-1$ small-talks occur between two consecutive regular  updates, see Fig.~\ref{fig:updates}.  The local states $\xi_i$ for all $i \in \mathcal{N}_1$ are collected in a new vector $\xi_{\mathcal{N}_1} \in \mathbb{R}^{N_1} $. Since  only oscillators in $\mathcal{N}_1$ can exchange their information, the state update according to  \eqref{eq:globup} in time  $t_{l} \in \mathcal{T}_l$ can be expressed by:
\begin{align}
&\omega_i(t^+_{l,k}) := \omega_i(t^-_{l,k}) \label{eq:localup1} \\
&\xi_{\mathcal{N}_1} (t^+_{l,k}) := \xi_{\mathcal{N}_1} (t^-_{l,k}) - \varphi P L_{11}\xi_{\mathcal{N}_1} (t^-_{l,k}) \label{eq:localup2}\\
&\xi_i(t^+_{l,k}) := \xi_i(t^-_{l,k}), ~i \in \mathcal{N}_2  \label{eq:localup3}
\end{align}
with $P= I_{N_1}-\frac{1}{N_1}1_{N_1}1_{N_1}^T$. With these preparations, the problem to be investigated in this paper is now defined as follows:
  \begin{definition}{(Problem Definition)}\label{des:def1}
Given the networked oscillators with continuous dynamics in \eqref{eq:dyn}, regular updates in  \eqref{eq:globup}, and small-talk updates in \eqref{eq:localup1} - \eqref{eq:localup3}, the task is to investigate how the intervals  $\Delta T$  and  $\tau$ of the regular and small-talk updates jointly influence the synchronization error. In particular, given a bounded neighborhood $\Omega$ of the synchronized behavior of the oscillator network, the goal is to determine conditions for $\tau$ and $\Delta T$ such that the oscillator trajectories remain in $\Omega$ for all future time if they are initialized in $\Omega$.
\end{definition}

When applying singular perturbation analysis to the problem at hand, a second perturbation parameter relating to $\tau$ is required in addition to $\Delta T$. The presence of two perturbation parameters complicates the evaluation of how $\Delta T$ and $\tau$ jointly affect the bound of the synchronization error, or conversely, how the maximally permissible values of $\Delta T$ and $\tau$ for a given  error bound depend on each other. Therefore,  an alternative approach (to perturbation analysis) is adopted in the following section.

\section{PRACTICAL SYNCHRONIZATION UNDER ADDITIONAL SMALL-TALK UPDATES}

\subsection{Transformation of the  Liénard Oscillator}

To solve the problem in  Def.~\ref{des:def1}, a  matrix 
$R \in \mathbb{R}^{N \times (N-1)}$ satisfying:
\begin{align} \label{eq:requirementR}
R^\top R = I_{N-1}, ~~1_N^\top R = 0_{1 \times (N-1)}
\end{align}
is introduced first. Therein, $1_N$ denotes the column vector containing only ones. The local states $\xi_i$ for all $i \in \mathcal{N}$ are collected in a vector $\xi_{\mathcal{N}} \in \mathbb{R}^{N}$, and:
\begin{align} 
&s := \frac{1}{N} \sum_{i=1}^N \xi_i,~~s \in \mathbb{R} \label{eq:transform2} \\
& \zeta := R^\top \xi,~~\zeta \in \mathbb{R}^{N-1} \label{eq:transform3} \\
 & \xi_{\mathcal{N}} = s 1_N + R \zeta. \label{eq:transform1}
\end{align}
As stated in \cite{tanwani2021lyapunov}, the scalar $s$ denotes the average of all  $\xi_i$, while
the  vector $\zeta$ represents the disagreement among  the local $\xi_i$, $i \in \mathcal{N}$ in the sense of deviations from the average (see the example in Sec.~\ref{example}).
Further define a vector  $\omega_{\mathcal{N}} := [\omega_1, \ldots, \omega_N]^T$ , its dynamics, as well as the dynamics  of  $s(t)$ and $ \zeta(t)$ based on \eqref{eq:dyn} as follows:
\begin{align}
&  \dot{\omega}_{\mathcal{N}}(t) =-\omega_{\mathcal{N}}(t)  +1_Ns(t)  + R\zeta(t) \label{eq:transformdynamic1} \\
& \dot{s}(t)= \hspace{-1mm}\frac{1}{N}1_N^T \bigg(\phi(\omega_{\mathcal{N}}(t) ,1_Ns(t) )\hspace{-1mm}  \notag \\
&~~~~~~~~~~~~~~~~+\hspace{-1mm} (I_N \hspace{-1mm}- \hspace{-1mm}\text{diag}(f(\omega_{\mathcal{N}}(t) )))R\zeta(t) \bigg)  \label{eq:transformdynamic2}\\
&   \dot{\zeta}(t) = R^T \phi(\omega_{\mathcal{N}}(t),1_Ns(t))  \notag \\
&~~~~~~~~~~~~~~~~+R^T(I_N -diag(f(\omega_{\mathcal{N}}(t) )))R \zeta(t) \label{eq:transformdynamic3}
\end{align}
with:
\begin{gather}
    \phi(\omega_{\mathcal{N}}, \xi) = \begin{bmatrix}
        \phi_1(\omega_1, \xi_1) \\
        \vdots \\
        \phi_N(\omega_N, \xi_N)
    \end{bmatrix}, ~
    f(\omega_{\mathcal{N}}) = \begin{bmatrix}
        f_1(\omega_1) \\ \vdots \\ f_N(\omega_N)
    \end{bmatrix}. 
\end{gather}
For abbreviation, the dynamics of  $\omega_{\mathcal{N}}(t)$ and $s(t)$  in \eqref{eq:transformdynamic1} and \eqref{eq:transformdynamic2} 
is also referred to by:
\begin{align}\label{eq:transformdynamicbrevity}
\begin{bmatrix}  \dot{\omega}_{\mathcal{N}}(t) \\ \dot{s}(t) 
  \end{bmatrix} = F(\omega_{\mathcal{N}}(t),s(t), \zeta(t)).
\end{align}
Based on \eqref{eq:transform2} and \eqref{eq:transform3}, the following equations with $\Lambda_g := R^T W R$ hold true 
for the regular update at each time  $t_{g,k} \in \mathcal{T}_g$:
\begin{align} \label{eq:global_jumps_overall_network}
    & 
    \begin{cases} 
    %    \omega(t^+_{g,k}) = \omega(t^-_{g,k}) \\
        s(t^+_{g,k}) = s(t^-_{g,k}), \\
        \zeta(t^+_{g,k}) = \Lambda_g \zeta(t^-_{g,k}). 
    \end{cases} 
\end{align}
The coupling gain $\varphi$ in \eqref{eq:W} is chosen such that $0 <\varphi <\frac{1}{\lambda_N(L)}$ to ensure that the matrix $\Lambda_g$ is positive definite and  Schur stable.
For the values of $s(t_l)$ and $\zeta(t_l)$ for  small-talk updates in $t_{l} \in \mathcal{T}_l$, the following lemma is formulated:
\begin{lemma} \label{l:localup}
Suppose that  $L_ {11}$ in \eqref{eq:laplaciandecompse} satisfies  $\sum_{j=1}^{N_1} [L_{11}]_{ij} = c$ for a $c \in \mathbb{R}$ and for all $i \in \{1, ..., N_1\}$. Then,  the  following equations hold  for each  $t_{l} \in \mathcal{T}_l$:
 \begin{align} \label{eq:local_jumps_overall_network}
    \begin{cases}
        s(t^+_{l,k}) = s(t^-_{l,k}) \\
        \zeta(t^+_{l,k}) = \Lambda_l \zeta(t^-_{l,k})
    \end{cases} 
    \end{align} 
 with:
  \begin{align}
  \Lambda_l := (I_{N-1} - \varphi R^T \begin{bmatrix} P L_{11} & 0_{N_1 \times N_2} \\ 0_{N_2 \times N_1} & 0_{N_2 \times N_2} \end{bmatrix} R). \label{eq:lambdaL} \\
  \notag
     \end{align} 
\end{lemma}

This lemma  follows from the update rules in \eqref{eq:localup1} to \eqref{eq:localup3} and the fact that $PL_{11} = L_{11}P$ exists (i.e., the matrices  $P$ and $L_{11}$ commute) if  $\sum_{j=1}^{N_1} [L_{11}]_{ij} = c$ holds for all $i \in \{1, ..., N_1\}$.

\begin{lemma} \label{l:LambdaL}
For $0 < \varphi < \frac{1}{\lambda_{N}(L)}$, the matrix $\Lambda_l$ defined in \eqref{eq:lambdaL} is orthogonally diagonalizable and has $N_2$ eigenvalues equal to one, while the remaining 
$N_1-1$ eigenvalues lie strictly within the open interval  $(0,1)$. 
\end{lemma}

The proof to Lemma~\ref{l:LambdaL} is stated in the Appendix.

%%%%%%%%%%%%%%%%%%%%%%%%%%%%%%%%%%%%%%%%%%%%%%%%%%%%%%%%%%%%%%%%%%%
\subsection{Evolution of  $\begin{bmatrix}
        \omega_{\mathcal{N}}(t) \\ s(t)
    \end{bmatrix}$ for  $\zeta(t) =0$}

Let the averaged dynamics \eqref{eq:avgdynamics} satsify:
\begin{assumption}\label{assmp:avg}
  The averaged dynamics in \eqref{eq:avgdynamics}  has an exponentially stable limit cycle  $\mathcal{A}_0$.
\end{assumption}
Under this assumption, the following result is provided:
\begin{lemma} \label{l:highLC}
Suppose that $\zeta(t) =0$ holds  for all $t \ge t_0$ in \eqref{eq:transformdynamicbrevity}. Then,  the resulting dynamics:
  \begin{align}\label{l:highLCzeta0}
    \begin{bmatrix}
        \dot{\omega}_{\mathcal{N}}(t) \\ \dot{s}(t)
    \end{bmatrix} =F(\omega_{\mathcal{N}}(t),s(t), 0).
  \end{align}
has an exponentially stable limit cycle $\mathcal{A} \in \mathbb{R}^{N+1}$ and satisfies:
  \begin{align}\label{eq:globalLC}
\mathcal{A}= \{ (\chi,\ldots, \chi, s ) \in \mathbb{R}^{N+1}: (\chi, s) \in \mathcal{A}_0 \}.\\
\notag
     \end{align}
\end{lemma}

Note that in  \cite{tanwani2021lyapunov}, the authors only claim that $\mathcal{A}$ is asymptotically stable, whereas  the proof showing that $\mathcal{A}$ is also   exponentially stable is provided in the Appendix.  An implication of Lemma~\ref{l:highLC} is that for  $\zeta(t)=0$  and if the local states $\omega_i(t)$, $i \in \mathcal{N}$ and $s(t)$, are all initialized to a neighborhood of  $\mathcal{A}$, then they will exponentially converge to  $\mathcal{A}$, i.e.  $\underset{t \to \infty}{\text{lim sup}} |(\omega_{\mathcal{N}}(t),  s(t))|_{\mathcal{A}} =0$. Furthermore, based on the definition of  $\mathcal{A}$ in \eqref{eq:globalLC}, it holds that 
$\underset{t \to \infty}{\text{lim sup}} |    (w_i(t), \xi_i(t) ) |_{\mathcal{A}_0}  =0$ for all  $i \in \mathcal{N}$.

Since the limit cycle $\mathcal{A}$ is  exponentially stable, the converse theorem  \cite{hauser1994converse} implies the existence of a Lyapunov function $V: \mathbb{R}^{N} \times \mathbb{R} \to \mathbb{R}_{\ge 0}$ satisfying:
  \begin{align}
& k_1 |(\omega_{\mathcal{N}},  s)|^2_{\mathcal{A}} \leq V(\omega_{\mathcal{N}},  s)  \leq  k_2 |(\omega_{\mathcal{N}},  s)|^2_{\mathcal{A}} \label{eq:lya1} \\
& \dot{V}(\omega_{\mathcal{N}},  s) _{\text{along \eqref{l:highLCzeta0}}} \leq -k_3 |(\omega_{\mathcal{N}},  s)|^2_{\mathcal{A}} \label{eq:lya2} \\
&\|\nabla V(\omega_{\mathcal{N}},  s) \| \leq k_4 |(\omega_{\mathcal{N}},  s)|_{\mathcal{A}} \label{eq:lya3}
       \end{align}
for all vectors $ \begin{bmatrix}
        \omega_{\mathcal{N}} \\ s
    \end{bmatrix}$ located in
a closed tubular neighborhood around the limit cycle $\mathcal{A}$:
\begin{equation}\label{l:tubular}
    \Omega = \{    \begin{bmatrix}
        \omega_{\mathcal{N}} \\ s
    \end{bmatrix} \in \mathbb{R}^{N+1}:  |(\omega_{\mathcal{N}},  s)|_{\mathcal{A}} \leq r\},~~r>0
\end{equation}

%%%%%%%%%%%%%%%%%%%%%%%%%%%%%%%%%%%%%%%%%%%%%%%%%%%%%%%%%%%%
\subsection{Evolution of  $\begin{bmatrix}
        \omega_{\mathcal{N}}(t) \\ s(t)
    \end{bmatrix}$ for  $\zeta(t)  \ne 0$}

For $\zeta(t)  \ne 0$ in   \eqref{eq:transformdynamicbrevity}, i.e., there exists disagreement among  the local $\xi_i(t)$, $i \in \mathcal{N}$,
the following lemma describes how $\begin{bmatrix}
        \omega_{\mathcal{N}}(t) \\ s(t)
    \end{bmatrix}$ evolves over time when it is initialized   near  the  limit cycle $\mathcal{A}$ of \eqref{l:highLCzeta0}:

\begin{lemma} \label{l:preserve}
Given the set $ \Omega $ and the constant $r >0$ in \eqref{l:tubular},  let a constant $  \alpha$ be determined by:
\begin{gather} \label{eq:alpha}
    \alpha := \frac{k_1 k_3}{k_2 k_4 \sup_{(\omega,s) \in \Omega} \|\frac{\partial F(\omega,s,0)}{\partial \zeta}\|}
\end{gather}
with constants $k_1, k_2, k_3$, and $k_4$ from \eqref{eq:lya1} to \eqref{eq:lya3}. If then $\begin{bmatrix}
        \omega_{\mathcal{N}}(t) \\ s(t)
    \end{bmatrix}$  satisfies:
      \begin{align}\label{eq:omegainitial}
|(\omega(t_0),s(t_0))|_{\mathcal{A}} \leq \sqrt{\frac{k_1}{k_2}}  r
      \end{align}
    for $t_0$   and if  $\zeta(t)$ satisfies:
  \begin{align}\label{eq:zetabound}
||\zeta(t)|| \le \zeta_{max} \le \alpha r
      \end{align}
      for $t \ge t_0$, the relation:
         \begin{align}\label{eq:omegaevolution}
 |(\omega_{\mathcal{N}}(t),  s(t))|_{\mathcal{A}} \leq r
      \end{align}   
holds  for  \eqref{eq:transformdynamicbrevity}  with  $t \ge t_0$.  
\end{lemma}
\begin{proof}
     The time derivative of the Lyapunov function $V$ in \eqref{eq:lya1} along   \eqref{eq:transformdynamicbrevity} satisfies: 
    \begin{align*}
         &   \dot{V} \hspace{-1mm} = \hspace{-1mm} \dot{V}(\omega_{\mathcal{N}},\hspace{-0.3mm}s)|_{\text{along   \eqref{eq:transformdynamicbrevity}}} \hspace{-1mm} + \hspace{-1mm}\nabla \hspace{-0.3mm}V\hspace{-0.3mm}(\omega_{\mathcal{N}},\hspace{-0.3mm}s)^T\hspace{-0.7mm} (F(\omega_{\mathcal{N}},\hspace{-0.3mm}s,\hspace{-0.3mm}\zeta)\hspace{-1mm}- \hspace{-1mm} F(\omega,\hspace{-0.3mm}s,\hspace{-0.3mm}0)) \\
            &\hspace{-1mm}  \leq \hspace{-1mm}  -k_3|(\omega_{\mathcal{N}},s)|_{\mathcal{A}}^2 \hspace{-1mm} +\hspace{-1mm}  k_4|(\omega_{\mathcal{N}},s)|_\mathcal{A}\| F(\omega_{\mathcal{N}},s,\zeta)\hspace{-1mm}  -\hspace{-1mm}  F(\omega_{\mathcal{N}},s,0)\| \\
           & \leq \hspace{-1mm}  -k_3|(\omega_{\mathcal{N}},s)|_{\mathcal{A}}^2 \hspace{-1mm} + \hspace{-1mm}  k_4|(\omega_{\mathcal{N}},s)|_\mathcal{A} \hspace{-3mm} \sup\limits_{(\omega_{\mathcal{N}},s) \in \Omega} \hspace{-3mm}  \| \frac{\partial F(\omega_{\mathcal{N}},s,0)}{\partial \zeta} \| \zeta_{max} \\
          & \leq - \frac{k_3}{k_2} V + \frac{k_4}{\sqrt{k_1}} \sqrt{V} \hspace{-1mm}\sup_{(\omega_{\mathcal{N}},s) \in \Omega} \hspace{-1mm}\| \frac{\partial F(\omega_{\mathcal{N}},s,0)}{\partial \zeta} \| \zeta_{max}, 
    \end{align*}
   where  the first inequality is due  to \eqref{eq:lya2}  and \eqref{eq:lya3},  the second inequality is obtained by using  the Mean-Value Theorem, and the last inequality is due to  \eqref{eq:lya1}. The Grönwall Lemma implies:
    \begin{align*}
& |(\omega_{\mathcal{N}}(t),  s(t))|_{\mathcal{A}} \le e^{-\frac{k_3}{2k_2}(t-t_0)} \sqrt{\frac{k_2}{k_1}} |(\omega(t_0),s(t_0))|_{\mathcal{A}}\\
&+ \frac{k_2 k_4}{k_1 k_3} \hspace{-1mm}\sup_{(\omega_{\mathcal{N}},s) \in \Omega} \hspace{-1mm} \|\frac{\partial F(\omega,s,0_{N-1})}{\partial \zeta}\|  (1 - e^{-\frac{k_3}{2k_2}(t-t_0)})\zeta_{max}.
   \end{align*}
 Thus, if \eqref{eq:omegainitial} and \eqref{eq:zetabound} are satisfied, the inequality above further yields:
       \begin{align*} 
 |(\omega_{\mathcal{N}}(t),  s(t))|_{\mathcal{A}} \le r 
   \end{align*}
  for all $t \ge t_0$ and thus completes the proof. 
\end{proof}
Lemma~\ref{l:preserve} states that if $\|\zeta(t)  \|$ does not exceed the threshold $\zeta_{max}$, then the set $\Omega$ remains invariant under the dynamics \eqref{eq:transformdynamicbrevity}. In other words, each oscillator $i \in \mathcal{N}$ continues to oscillate within a bounded neighborhood of the limit cycle $\mathcal{A}_0$ of the averaged dynamics \eqref{eq:avgdynamics}. Furthermore, the proof reveals that a smaller $\zeta_{max}$ leads to a smaller ultimate bound $\underset{t \to \infty}{\text{lim sup}} |    (\omega_{\mathcal{N}}(t),  s(t))|_{\mathcal{A}}$, as established by Theorem 9.1 in \cite{khalil2002nonlinear}. However, the dynamics of $\zeta(t)$ in \eqref{eq:transformdynamic3} does not guarantee that $\|\zeta(t)  \|$  always remain below  $\zeta_{max}$. Therefore, the next section investigates how the regular and small-talk updates help to ensure that  $\|\zeta(t)  \| \le \zeta_{max}$ for all $t \ge t_0$.

%%%%%%%%%%%%%%%%%%%%%%%%%%%%%%%%%%%%%%%%%%%%%%%%%%%%%%%%%%%%
\subsection{Evolution  of $\|\zeta(t) \|$ for Regular and Small-Talk Updates}

For the matrix $\Lambda_l$ in \eqref{eq:lambdaL}, let $V= [V_s, V_u] \in \mathbb{R}^{(N-1) \times (N-1)}$ be an orthogonal  eigenvector matrix with  $V_s$ containing the eigenvectors corresponding to the $N_1-1$ eigenvalues lying within $(0,1)$.  $V_u$ contains the eigenvectors  corresponding to the $N_2$ eigenvalues equal to one. It holds true that:
           \begin{align} \label{eq:sumpropertyinitial}
[V_s, V_u]^T \Lambda_l [V_s, V_u] = \text{diag}(J_s, I_{N_2}), ~J_s \in \mathbb{R}^{(N_1-1) \times (N_1-1)}
   \end{align}
and: 
\[
[V_s, V_u]^T =[V_s, V_u]^{-1}
\]
since  $\Lambda_l$ is symmetric.  By applying a   similarity transformation based on $V$ to  $\zeta(t)$:
\begin{gather} \label{eq:trafo_mu_zeta}
        \rho(t) = 
        \begin{bmatrix}
            \rho_s(t) \\ \rho_u(t)    
        \end{bmatrix} = [V_s, V_u]^T\zeta(t),
    \end{gather}
the relations:
           \begin{align} \label{eq:sumpropertyaddition}
V_s\rho_s(t) + V_u\rho_u(t) = \zeta(t)
   \end{align}
and: 
           \begin{align} \label{eq:sumproperty}
    \|\rho_s(t)\|^2 + \|\rho_u(t)\|^2 = \|\zeta(t)\|^2
   \end{align}
   always hold.
This similarity transformation allows one to treat the part $\rho_s(t)$ affected by the small-talk update in  \eqref{eq:local_jumps_overall_network} and the part  $\rho_u(t)$  (which remains unaffected) separately.
Define $z_s(t):= \|\rho_s (t)\|$ and $z_u(t):= \|\rho_u (t)\|$ with:
           \begin{align} \label{eq:sumpropertynew}
    z^2_s(t) + z^2_u(t) = \|\zeta(t)\|^2,~~~\|  \begin{bmatrix} z_s(t) \\ z_u(t) \end{bmatrix} \|= \|\zeta(t)\|
   \end{align}
according to \eqref{eq:sumproperty}.  The vector $[z_s(t), z_u(t)]^T$ satisfies:
 \begin{gather} \label{eq:local_jumps_Z}
        \begin{bmatrix}
            z_s(t^+_{l,k}) \\
            z_u(t^+_{l,k})
        \end{bmatrix} =     \begin{bmatrix}
            \|J_s \rho_s(t^-_{l,k})\| \\
            z_u(t^-_{l,k})
        \end{bmatrix}  \leq 
        \underbrace{\begin{bmatrix}
            \|J_s\| & 0 \\ 0 & 1  
        \end{bmatrix}}_{:= \Gamma_l}    \begin{bmatrix}
            z_s(t^-_{l,k})\\
            z_u(t^-_{l,k})
        \end{bmatrix}
    \end{gather} 
in each  $t_{l} \in \mathcal{T}_l$ of small-talk updates, where $  \|\Gamma_l\| = 1$ and  $  \|J_s\| < 1$ due to \eqref{eq:sumpropertyinitial}. At  each time  $t_{g} \in \mathcal{T}_g$ of regular update,  the inequality:
    \begin{gather} \label{eq:global_jumps_Z}
              \begin{bmatrix}
            z_s(t^+_{g,k})\\
            z_u(t^+_{g,k})
        \end{bmatrix}\leq \underbrace{\begin{bmatrix}\|V_s^T \Lambda_g V_s\|     & \|V_s^T \Lambda_g V_u\| \\ \|V_u^T \Lambda_g V_s\| & \|V_u^T \Lambda_g V_u\| \end{bmatrix}}_{:= \Gamma_g}      \begin{bmatrix}
            z_s(t^-_{g,k}) \\
            z_u(t^-_{g,k})
        \end{bmatrix}
    \end{gather}
    can be obtained in a similar way, together with the following property for the matrix $\Gamma_g$:
\begin{lemma} \label{l:GammaG}
The matrix $\Gamma_g$  in \eqref{eq:global_jumps_Z} is Schur stable and satisfies $\|\Gamma_g\| < 1$.
\end{lemma}

The  proof to this lemma is stated in Appendix.  Based on the equations in \eqref{eq:sumpropertynew}, the  evolution of $\|\zeta(t)\|$ over time can be equivalently  analyzed by examining  the evolution of $\| \begin{bmatrix} z_s(t) \\ z_u(t) \end{bmatrix} \|$. In the first step, the time derivative of  $z_s(t)$ is given by:
    \begin{align}\label{eq:flowrhos}
& \dot{z}_s (t) =  \frac{\rho^T_s (t) \dot{\rho}_s (t)}{z_s(t)}=  \frac{\rho^T_s (t) V_s^{T} \dot{\zeta}(t)}{z_s(t)} \notag\\
& =   \frac{\rho^T_s (t) V_s^{T} }{z_s(t)} \bigg( R^T \phi(\omega_{\mathcal{N}}(t),1_Ns(t))  \notag\\
&\hspace{-1mm} + \hspace{-1mm}  R^T(I_N\hspace{-1mm}  - \hspace{-1mm} \text{diag}(f(\omega_{\mathcal{N}}(t) )))R ( V_s\rho_s(t)\hspace{-1mm}  +\hspace{-1mm}  V_u\rho_u(t)) \bigg)
   \end{align}
where the last equality is obtained from  \eqref{eq:trafo_mu_zeta} and the dynamics of $\zeta(t)$. The three terms on the right-hand side of \eqref{eq:flowrhos} are upper-bounded by:
    \begin{align}\label{eq:Z_s_1st_fraction_bound}
&\frac{\rho^T_s (t) V_s^T R^T  \phi(\omega_{\mathcal{N}}(t),1_Ns(t))}{z_s(t)} \notag \\
&\leq \underbrace{\sup_{(\omega_{\mathcal{N}}(t),s(t))\in \Omega} \|V_s^T R^T  \phi(\omega_{\mathcal{N}}(t),1_Ns(t))\|}_{:= b_s}
   \end{align}
   and:
\begin{gather} \label{eq:Z_s_2nd_fraction_bound}
    \begin{split}
  & \frac{\rho^T_s (t) V_s^T R^T (I_N -\text{diag}(f(\omega_{\mathcal{N}}(t) )))R  V_s\rho_s(t)}{z_s(t)} \leq a_{ss} z_s(t)
    \end{split}
\end{gather}
   with  $a_{ss}$  denoting the largest magnitude of all eigenvalues of $V_s^T R^T (I_N -diag(f(\omega_{\mathcal{N}}(t) )))R  V_s$ for all $\begin{bmatrix}
        \omega_{\mathcal{N}}(t) \\ s(t)
    \end{bmatrix} \in     \Omega$, as well as: 
\begin{gather} \label{eq:Z_s_3rd_fraction_bound}
    \begin{split}
    &  \frac{\rho^T_s (t) V_s^T R^T (I_N -diag(f(\omega_{\mathcal{N}}(t) )))R  V_u\rho_u(t)}{z_s(t)}   \\
    & \leq \| V_s^T R^T (I_N -\text{diag}(f(\omega_{\mathcal{N}}(t) )))R  V_u\rho_u(t)\| \\ 
    &\leq \underbrace{ \hspace{-2mm}\sup_{(\omega_{\mathcal{N}}(t),s(t))\in \Omega}\hspace{-2mm}\|V_s^T R^T (I_N -\text{diag}(f(\omega_{\mathcal{N}}(t) )))R  V_u\|}_{:=a_{su}} \rho_u(t).
    \end{split}
\end{gather}
By applying a similar bounding procedure to the dynamics of $z_u(t)$, the following inequalities can by set up by use of non-negative constants $b_u$, $a_{uu}$ and  $a_{us}$:
     \begin{gather} \label{eq:Zdyn}
      \begin{bmatrix}  \dot{z}_s (t) \\  \dot{z}_u (t) \end{bmatrix} \leq \underbrace{\begin{bmatrix} a_{ss} & a_{su} \\ a_{us} & a_{uu}
        \end{bmatrix}}_{:=A}   \begin{bmatrix}  z_s (t) \\  z_u (t) \end{bmatrix} + \underbrace{\begin{bmatrix} b_s \\ b_u \end{bmatrix}}_{:=B} 
    \end{gather}
Here, only non-negative elements are contained in $A$ and $B$. It is assumed that  $A^{-1}$ exists, as one can always increase the value of each element in $A$ without affecting the  inequality.
Now the main result is provided: 
\begin{theorem}
Given  $\zeta_{max}$   in \eqref{eq:zetabound} and the regular update interval $\Delta T$,    assume that
$|(\omega(t^+_0),s(t^+_0))|_{\mathcal{A}} \leq \sqrt{\frac{k_1}{k_2}}  r$ and  $\|\zeta(t^+_0)\| \leq \|\Gamma_g\| \zeta_{max}$ hold at the initial time $t^+_0$. If the interval $\tau$ of the small-talk update and the corresponding constant $\mu = \frac{\Delta T}{\tau}$ in \eqref{eq:smalltalktime} satisfy:
    \begin{align}\label{eq:finaltheorem1}
        \| \tilde{A}_{m} \| \zeta_{max} + \|\tilde{B}_{m}\| \leq \zeta_{max} 
    \end{align}
for all  $m \in \{1, ..., \mu\}$ with:
    \begin{align}
          & \tilde{A}_m = e^{A\tau} (\Gamma_l e^{A\tau})^{m-1},\label{eq:finaltheorem2}\\
        & \tilde{B}_m = (\sum_{j=1}^{m-1} (e^{A\tau}\Gamma_l)^{j-1})A^{-1}(e^{A\tau}-I_2)B,    \label{eq:finaltheorem3}
    \end{align}
then it holds for \eqref{eq:transformdynamicbrevity}  and $t >t_0$ that:
         \begin{align}\label{eq:finalproerpty}
 |(\omega_{\mathcal{N}}(t),  s(t))|_{\mathcal{A}} \le r. 
      \end{align}   
\end{theorem}

\begin{proof}
Consider the interval $t \in (t_0, t_0 + \tau)$, spanning from the initial time to the first small-talk update. \eqref{eq:Zdyn} implies:
   \begin{gather}
    \begin{split}
        \begin{split}
         \begin{bmatrix}  z_s (t) \\  z_u (t) \end{bmatrix} \hspace{-0.6mm} \leq \hspace{-0.6mm}e^{A(t-t^+_0)} \hspace{-0.6mm} \begin{bmatrix}  z_s (t^+_0) \\  z_u (t^+_0) \end{bmatrix} \hspace{-0.6mm}+ \hspace{-0.6mm}A^{-1}(e^{A(t-t^+_0)}-I_2)B \label{eq:Gronwall_inequaltiy}
        \end{split}
    \end{split}
    \end{gather}
    according to the   generalized Grönwall inequality (see Corollary 2 in \cite{chandra1976linear}).   Since all elements in  $A$ and $B$ are non-negative, it further applies that:
    \begin{align}\label{eq:positivedevelop}
             \begin{bmatrix}  z_s (t) \\  z_u (t) \end{bmatrix} \leq   e^{A \tau}  \begin{bmatrix}  z_s (t^+_0) \\  z_u (t^+_0) \end{bmatrix} +A^{-1}(e^{A\tau}-I_2)B
    \end{align}
 holds for $t \in  (t_0, t_0 + \tau)$.  Due to   \eqref{eq:sumpropertynew},  \eqref{eq:positivedevelop} further implies:
     \begin{align*}
             \|\begin{bmatrix}  z_s (t) \\  z_u (t) \end{bmatrix}\| & \hspace{-0.6mm} =  \hspace{-0.6mm} \|\zeta(t)\|\hspace{-0.6mm}  \leq \hspace{-0.6mm}  \| e^{A \tau} \|   \|\zeta(t^+_0)\| \hspace{-0.6mm}  + \hspace{-0.6mm}  \| A^{-1}(e^{A\tau}-I_2)B  \| \\
             & \leq \| \tilde{A}_{1} \| \|\Gamma_g\| \zeta_{max} + \|\tilde{B}_{1}\|  \leq \zeta_{max}
    \end{align*}
    according to Lemma~\ref{l:GammaG} and \eqref{eq:finaltheorem1}, which indicates that $ \|\zeta(t)\|$ remains below the threshold $\zeta_{max}$ throughout the entire interval $(t_0, t_0 + \tau)$.  At the time $t_{l,1}:=t^+_0 + \tau$ (the first instance at which a small-talk update occurs) \eqref{eq:local_jumps_Z} implies that:
   \begin{gather}
    \begin{split}
        \begin{split}
      \begin{bmatrix}
            z_s(t^+_{l,1}) \\
            z_u(t^+_{l,1})
        \end{bmatrix}  \hspace{-1.3mm} \le   \hspace{-1mm} \Gamma_l   \hspace{-1mm} \begin{bmatrix}  z_s (t^-_{l,1}) \\  z_u (t^-_{l,1}) \end{bmatrix}   \hspace{-1.3mm}  \leq   \hspace{-1mm} \Gamma_l   e^{A \tau}  \hspace{-1mm} \begin{bmatrix}  z_s (t^+_0) \\  z_u (t^+_0) \end{bmatrix} \hspace{-1.3mm} +  \hspace{-1mm} \Gamma_l  A^{-1}  \hspace{-0.4mm} (e^{A\tau} \hspace{-1.3mm}- \hspace{-1mm}I_2)B \label{eq:Gronwall_inequaltiy2}
        \end{split}
    \end{split}
    \end{gather}
holds, and thus:
     \begin{align*}
 \|\zeta(t^+_{l,1})\| \leq \|\Gamma_l\| \|\zeta(t^-_{l,1})\|  \leq  \zeta_{max}.
    \end{align*}
By recursively applying \eqref{eq:Gronwall_inequaltiy} and \eqref{eq:Gronwall_inequaltiy2} for each successive interval $[t_{l,k}, t_{l,k+1})$ one obtains:
 \begin{align} \label{eq:Gronwall_inequaltiy3}
\begin{bmatrix}  z_s (t^-_{l,k+1}) \\  z_u (t^-_{l,k+1}) \end{bmatrix} \leq \tilde{A}_k  \begin{bmatrix}  z_s (t^+_0) \\  z_u (t^+_0) \end{bmatrix} + \tilde{B}_k
    \end{align}
for all $k=0,1, \ldots, \mu-1$, and thus:
     \begin{align*}
 \|\zeta(t^-_{l,k+1})\| \leq   \|\tilde{A}_k \|   \|\zeta(t^+_0)\|  + \|\tilde{b}_k \| \leq  \zeta_{max}
    \end{align*}
according to \eqref{eq:finaltheorem1}. In $t_{g,1}:=t^+_0 + \mu\tau$, the moment of the first regular update:
     \begin{align} \label{eq:Gronwall_inequaltiy5}
 \|\zeta(t^+_{g,1})\| &\leq  \|\Gamma_g\| \|\zeta(t^-_{g,1})\| =  \|\Gamma_g\| \|\zeta(t^-_{l,\mu})\| \leq   \|\Gamma_g\| \zeta_{max}
    \end{align}
        follows from \eqref{eq:global_jumps_Z}. As a result, the relation $ \|\zeta(t^+_{g,1})\| \le \|\Gamma_g\| \zeta_{max}$ holds  at time $t^+_{g,1}$, just as it does at the initial time $ t^+_0$. By recursively applying this analysis to each  successive  intervals $[t_{g,k}, t_{g,k+1})$, it follows that $ \|\zeta(t)\| \le  \zeta_{max}$  is satisfied for all $t > t_0$. Based on Lemma~\ref{l:preserve}, the relation in \eqref{eq:finalproerpty} must therefore hold   for all $t > t_0$.  
        \end{proof}

\section{NUMERIC EXAMPLE}\label{example}

This section demonstrates two key effects of including small-talk updates. First, when maintaining the same regular update interval as in the case without small-talk updates, small-talk updates yield a reduction in the norm of the maximal disagreement. Second, small-talk updates enable an extension of the regular update interval while preserving an equivalent norm of the maximal disagreement, thereby reducing the overall communication overhead. \\
To validate these points, a network of four Van-der-Pol oscillators is considered, governed by the \eqref{eq:dyn} with $f_i(\omega_i) = \kappa_i(\omega_i^2 - 1)$ and $g_i(\omega_i) = \omega_i$ as well as damping parameters  $\kappa_1 = 2.5$, $\kappa_2 = 2$, $\kappa_3 = 1.5$, and $\kappa_4 = 1$. The interaction graph is depicted in Fig.~\ref{fig:disagree3}. The Laplacian takes the form stated in \eqref{eq:laplaciandecompse} and is given by:
{
\setlength{\arraycolsep}{2pt}
\renewcommand{\arraystretch}{0.8}
\begin{gather*}
L = 
\begin{bmatrix}
2 & -1 & 0 & -1 \\
-1 & 2 & -1 & 0 \\
0 & -1 & 2 & -1 \\
-1 & 0 & -1 & 2
\end{bmatrix}, \quad L_{11} = 
\begin{bmatrix}
2 & -1 \\
-1 & 2 
\end{bmatrix}.
\end{gather*}
}
Note that the similarity transformation of the nodes $2$ and $4$ leaves $L$ unchanged due to the symmetric structure of the ring graph. Subgraph $G_1$ satisfies the requirements of Lemma~\ref{l:localup}. Since the damping parameters $\kappa_i$ of oscillators $1$ and $4$ exhibit the greatest mutual deviation, this subgraph selection is expected to yield the largest reduction in the synchronization error, compared to all other possible subgraphs of two oscillators. The matrix $R$ can be constructed via the Gram-Schmidt procedure to:
{
\setlength{\arraycolsep}{2pt}
\renewcommand{\arraystretch}{0.8}
\begin{gather*}
R =
\begin{bmatrix}
\frac{1}{\sqrt{2}} & \frac{1}{\sqrt{6}} & \frac{1}{\sqrt{12}} \\
-\frac{1}{\sqrt{2}} & \frac{1}{\sqrt{6}} & \frac{1}{\sqrt{12}} \\
0 & -\frac{2}{\sqrt{6}} & \frac{1}{\sqrt{12}} \\
0 & 0 & -\frac{3}{\sqrt{12}}
\end{bmatrix}.
\end{gather*}
}
Choosing $\varphi = 0.2$ fulfills $0<\varphi < \frac{1}{\lambda_N(L)}$ and results in the following matrices $\Lambda_g$ and $\Lambda_l$ for the regular and small-talk updates: 
{
\setlength{\arraycolsep}{2pt}
\renewcommand{\arraystretch}{0.8}
\begin{gather*}
\Lambda_g =
\begin{bmatrix}
0.4 & 0.1155 & -0.1633 \\
0.1155 & 0.5333 & 0.0943\\
-0.1633 & 0.0943 & 0.4667
\end{bmatrix}, \Lambda_l = \text{diag}(0.4,1,1).
\end{gather*}
}

Figure~\ref{fig:disagree} illustrates the two principal effects induced by small-talk updates:
Figure~\ref{fig:disagree1} demonstrates that introducing additional small-talk updates reduces the norm of the maximal disagreement, while the regular update interval is held constant. The corresponding phase portrait for the oscillator network under the small-talk settings of Fig.~\ref{fig:disagree1} is presented in Fig.~\ref{fig:disagree4}. The dependence of the peak disagreement norm on the small-talk update interval, for a fixed regular update time, is further investigated in Table~\ref{tab:convergence}. The results indicate that the peak of the disagreement norm converges to a positive lower bound as the small-talk update interval decreases, suggesting a fundamental limit on the achievable error reduction under the proposed update scheme.

The second effect is illustrated in Fig.~\ref{fig:disagree2}, where an improvement in communication efficiency is observed. Each regular update requires $8$ communication events, as any agent transmits its state to two neighbors, whereas each small-talk update requires only $2$ communication events. Over one regular update period in the small-talk setting ($\mu \cdot \tau = 0.3\,\mathrm{s}$ in this example), a total of $18$ communication events are performed. By contrast, achieving a comparable norm of the maximal disagreement using regular updates alone with $\Delta T = 0.1\,\mathrm{s}$ requires $24$ communication events within the same $0.3\,\mathrm{s}$ interval.

\begin{figure}[h!]
    \centering
    \begin{subfigure}{0.48\textwidth}
        \centering
        \psfrag{m1}[][b]{ $\kappa_1=2.5$}
        \psfrag{1}[][b]{ $1$}
        \psfrag{m2}[][b]{\small $\kappa_2=2$}
        \psfrag{2}[][b]{ $2$}
        \psfrag{m3}[][t]{$\kappa_3=1.5$}
        \psfrag{3}[][b]{ $3$}
        \psfrag{m4}[][t]{\small $\kappa_4=1$}
        \psfrag{4}[][b]{ $4$}
        \psfrag{G1}[][b]{\textcolor{red}{$G_1$}}
        \includegraphics[width=0.35\linewidth]{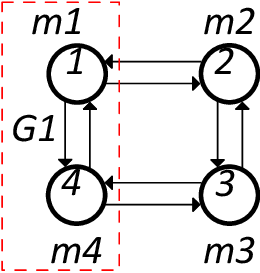}
        \caption{Used interaction network, where the subgraph containing the oscillators $1$ and $4$ use small-talk updates.}
        \label{fig:disagree3}
    \end{subfigure}
    \begin{subfigure}{0.48\textwidth}
        \centering
        \psfrag{t}[][b]{\scriptsize $t \hspace{1mm} \mathrm{[s]}$}
        \psfrag{z}[][t]{\scriptsize $||\zeta(t)||$}
        \includegraphics[width=0.65\linewidth]{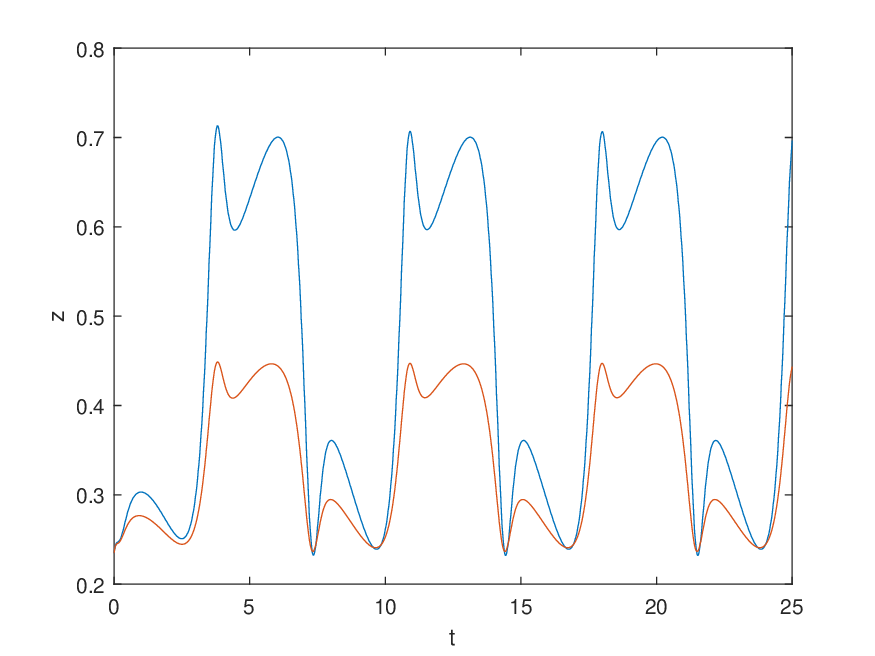}
        \caption{Comparison of the norm of the disagreement for two cases: Blue: only regular updates every $0.1 \mathrm{s}$. Orange: small-talk updates every $0.033\mathrm{s}$, $\mu=3$.}
        \label{fig:disagree1}
    \end{subfigure}
      \begin{subfigure}{0.48\textwidth}
        \centering
        \psfrag{x}[][b]{\scriptsize $\omega_i(t)$}
        \psfrag{v}[][t]{\scriptsize $\xi_i(t)$}
        \includegraphics[width=0.65\linewidth]{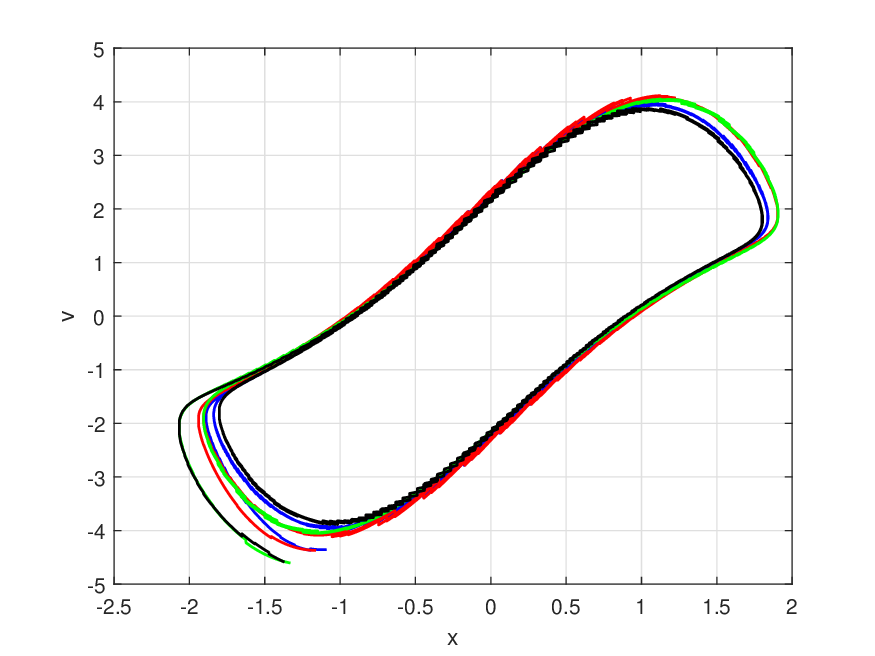}
        \caption{Phase portrait in the $\omega_i(t)$-$\xi_i(t)$ plane, when using small-talk updates with $\tau=0.033\mathrm{s}$ and $\mu=3$. Blue: Oscillator 1. Red: Oscillator 2. Green: Oscillator 3. Black: Oscillator 4.}
        \label{fig:disagree4}
    \end{subfigure}
    \hfill
    \begin{subfigure}{0.48\textwidth}
        \centering
        \psfrag{t}[][b]{\scriptsize $t \hspace{1mm} \mathrm{[s]}$}
        \psfrag{z}[][t]{\scriptsize $||\zeta(t)||$}
        \includegraphics[width=0.65\linewidth]{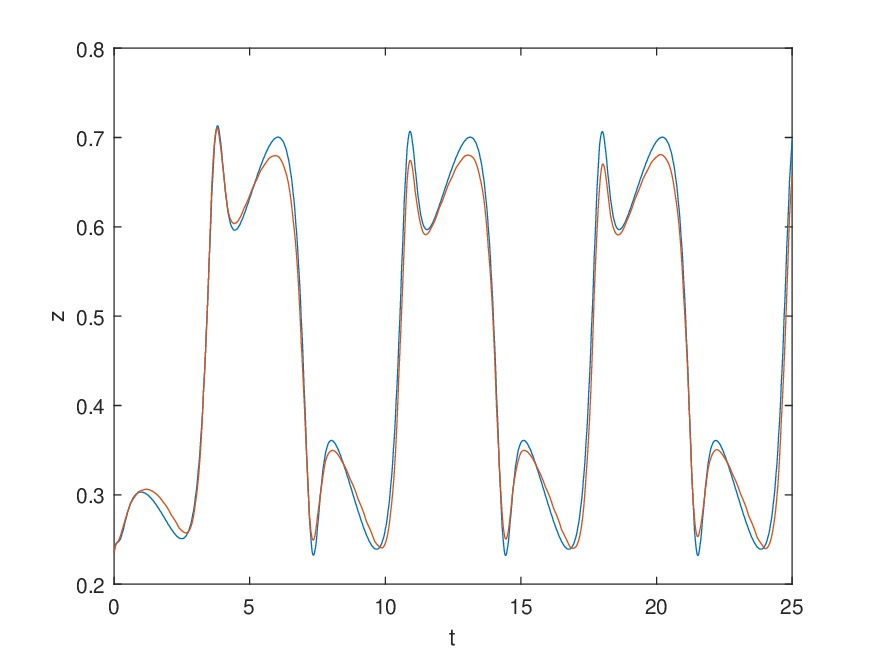}
        \caption{Comparison of the norm of the disagreement for two cases: Blue: Regular updates every $0.1 \mathrm{s}$. Orange: Small-talk updates every $0.05\mathrm{s}$, $\mu=6$.}
        \label{fig:disagree2}
    \end{subfigure}

    \caption{Analysis of the effect of small-talk updates. a) Used interaction network. b) Phase portrait of an example with small-talk updates. c) Additional small-talk updates reduce the maximal disagreement for the same regular update rate. d) Approximately same maximal disagreement, while small-talk updates reduce the communication by $25\%$. }
    \label{fig:disagree}
\end{figure}

\begin{table}[h]
\centering
\caption{Convergence of the peak of the disagreement norm $||\zeta(t)||$ with increasing number of local updates $\mu$ per global update interval ($t_g = 0.1$) for the same example as in Fig.~\ref{fig:disagree1}. The local update time is $\tau = \frac{t_g}{\mu}$. As $\mu$ increases, the maximal disagreement decreases and converges toward a lower bound.}
\label{tab:convergence}
\begin{tabular}{c||c|c|c|c|c|c|c|c}
$\mu$ & 1 & 2 & 4  & 8 & 10 & 15 & 20 \\
\hline
$\max(||\zeta(t)||)$ & 0.71 & 0.51 & 0.42 & 0.38 & 0.37 & 0.36 & 0.36 \\
\end{tabular}
\end{table}

\section{CONCLUSIONS}

This paper has investigated time-triggered practical synchronization of networked oscillators with two distinct communication frequencies and communication partners: In addition to regular updates involving the entire network, intermediate small-talk updates on a subgraph were introduced. This multi-rate update mechanism allows for maintaining a prescribed synchronization error while reducing communication effort or improving synchronization performance.

Sufficient conditions were derived to guarantee that an initial upper bounded distance to the synchronized behavior remains bounded over time. Based on the derived conditions the small-talk update interval and frequency, as well as the selection of the subgraph (on which the updates are performed) can be chosen.

Future work will focus on the optimal selection of subgraphs for small-talk updates in order to further improve performance and communication efficiency. Moreover, the extension to multiple subgraphs with heterogeneous update frequencies represents a promising direction for enhancing the flexibility and scalability of the proposed approach.


%%%%%%%%%%%%%%%%%%%%%%%%%%%%%%%%%%%%%%%%%%%%%%%%%%%%%%%%%%%%%%%%%%%%%%%%%%%%%%%%
\section*{APPENDIX}

\subsection{Proof of Lemma~\ref{l:LambdaL}}
\begin{proof}
    Since the matrices $P$ and $L_{11}$ commute and are symmetric, the matrix $\Lambda_l$ is symmetric and thus orthogonally diagonalizable according to the spectral theorem as per \cite{horn2012matrix}. \\
    Owing to $z^T L_{11} z = \sum_{i = 1}^{N_1}\sum_{j = 1}^{N_1} [A]_{ij} (z_i-z_j)^2 + \sum_{i = 1}^{N_1}\sum_{j = N_1+1}^{N} [A]_{ij} z_i^2$ for any $z \in \mathbb{R}^{N_1}$, the matrix $L_{11}$ is positive definite and can be rewritten using the Cholesky decomposition \cite{horn2012matrix}. With the fact that the matrix $P$ has one eigenvalue equal to zero and $N_1-1$ eigenvalues equal to one, Sylvester's law of inertia shows that the matrix product $PL_{11}$ has one eigenvalue equal to one and $N_1-1$ positive eigenvalues. From the similarity transformation involving the orthogonal matrix $[R, \frac{1}{\sqrt{N}1_N}]$, it becomes evident that $R^T diag(PL_{11}, 0) R$ contains the $N_{1}-1$ positive eigenvalues of $PL_{11}$ and has $N_2$ eigenvalues equal to zero. Therefore, the matrix $\Lambda_l = I_{N_1} - \varphi R^T diag(PL_{11}, 0)R$ has $N_2$ eigenvalues equal to one, while the remaining $N_1-1$ eigenvalues are of the form $1-\varphi \lambda_i(PL_{11})$, $i \in \{1, ..., N_1-1\}$. Due to $\lambda_{max} (PL_{11}) = \lambda_{max} (PL_{11}P) \leq \max_{z \neq 0} \frac{z^T L_{11} z}{\|x\|^2} = \lambda_{max} (L)$ and $\lambda_{max}(L_{11}) \leq \lambda_N(L)$ obtained by Cauchy's eigenvalue interlacing theorem according to \cite{horn2012matrix}, the remaining $N_1-1$ eigenvalues lie within $(0,1)$ if $0 < \varphi < \frac{1}{\lambda_{N}(L)}$ holds.       
\end{proof}

\subsection{Proof of Lemma~\ref{l:highLC}}
\begin{proof}
    Identically to $\xi$, the global state vector $\omega$ can be partitioned into an average component $\chi$ and a component $\tilde{\chi}$ describing the disagreement using the matrix $R$. With this in mind, the distance of an arbitrary point $P = (\chi, \tilde{\chi}, s)$ to the limit cycle $\mathcal{A}$ can be expressed as follows:
    \begin{gather}
    \begin{split}
        | (\omega, s) |_{\mathcal{A}}
        & = \inf_{P' = (\chi', \tilde{\chi}', s') \in \mathcal{A}} \|\vec{0P} - \vec{0P'}\|  \\
        & = \inf_{P' \in \mathcal{A}} \| \begin{bmatrix} 
            \chi -  \chi' \\  0 \\ s - s'
        \end{bmatrix} + 
        \begin{bmatrix}
            0 \\ \tilde{\chi} - \tilde{\chi}' \\ 0 \\
        \end{bmatrix} \|
    \end{split}
    \end{gather}
    Combining the fact that the vector $[\chi-\chi', 0^T, s-s']^T$ only lies in the $(\chi,s)$ coordinate space and the vector $[0, \tilde{\chi}-\tilde{\chi}',0]^T$ in the $\tilde{\chi}$ coordinate space with the property $1_N^T R = 0$, these coordinate spaces and thus the vectors $[\chi-\chi', 0^T, s-s']^T$ and $[0, \tilde{\chi}-\tilde{\chi}',0]^T$ are orthogonal. Taking this into account, the squared distance $|(\omega,s)|_\mathcal{A}^2$ can be written as: 
    \begin{gather} \label{eq:squared_dist}
    \begin{split} 
        |(\omega,s)|_\mathcal{A}^2 &= \inf_{P' \in \mathcal{A}}( \| \begin{bmatrix} \chi-\chi' \\ s-s' \end{bmatrix} \|^2 + \| \tilde{\chi}-\tilde{\chi}' \|^2) \\
        &= \inf_{(\chi',s')\in \mathcal{A_0}, \: \tilde{\chi}' = 0} (\| \begin{bmatrix} \chi-\chi' \\ s-s'\end{bmatrix}\|^2 + \|\tilde{\chi}-\tilde{\chi}'\|^2) \\
        &= |(\chi,s)|^2_{\mathcal{A}_0} + \|\chi\|^2.
    \end{split}
    \end{gather}
    For the dynamical system $F(\omega,s,0)$, define the Lyapunov function: $V_T: \Omega_T \to \mathbb{R}_{\geq 0}, V_T(\chi, \tilde{\chi},s) = W_T(\chi,s) + k\|\tilde{\chi}\|^2$ with: 
    \begin{gather}
        \Omega_T = \{\begin{bmatrix}
        s & \chi & \tilde{\chi}
        \end{bmatrix}^T \in \mathbb{R}^N: V_T(s,\chi,\tilde{\chi}) < \min(c, kC^2)\}
    \end{gather}
    and $W_T$ meeting the conditions: 
    \begin{align}
        &c_1 |(\chi,s)|_{\mathcal{A}_0} \leq W_T(\chi,s) \leq c_2 |(\chi,s)|_{\mathcal{A}_0} \\ 
        &\dot{W}_T (\chi,s)|_{\text{along } [\dot{\chi}, \dot{s}]^T} \leq -c_3 |(\chi,s)|_{\mathcal{A}_0} \\
        & \| \nabla W_T(\chi,s) \| \leq c_4 |(\chi,s)|_{\mathcal{A}_0}
    \end{align}
    on the set $\{\begin{bmatrix} \chi & s \end{bmatrix}^T \in \mathbb{R}^2: W_T(\chi,s) \leq c \}$, according to \cite{tanwani2021lyapunov} .
    Combining these properties with \eqref{eq:squared_dist}, it becomes evident that the Lyapunov function $V_T(\chi, \tilde{\chi}, s)$ satisfies: 
    \begin{align}
        & \min\{ c_1,k \} |(\omega,s)|_\mathcal{A}^2 
        \leq V_T \leq \max\{ c_2,k \} |(\omega,s)|_\mathcal{A}^2 \\
        & \dot{V}_T(\chi,\tilde{\chi},s) \leq  - \frac{\min\{\frac{c_3}{2},k\}}{\max\{ c_2,k\}} V_T(\chi, \tilde{\chi},s). \label{eq:bound_dot_VT}
    \end{align}
    Applying the comparison lemma  \cite{khalil2002nonlinear} to \eqref{eq:bound_dot_VT} yields: 
    \begin{gather}
        V_T \leq \exp(- \frac{\min\{ \frac{c_3}{2},k\}}{\max \{ c_2,k\}} (t-t_0)) V_T(\chi(t_0), \tilde{\chi}(t_0),s(t_0)).
    \end{gather}
    This inequality is equivalent to 
    \begin{gather}
    \begin{split}
        &|(\omega,s)|_\mathcal{A} \leq \\ &\exp(- \frac{\min\{ \frac{c_3}{2},k\}}{2\max \{ c_2,k \} } (t-t_0)) \sqrt{\frac{\max \{ c_2,k\}}{\min \{ c_1,k\}}} |(\omega(t_0),s(t_0)|_\mathcal{A},
    \end{split}
    \end{gather}
    what implies that $\mathcal{A}$ as per \eqref{eq:globalLC} is locally exponentially stable.    
\end{proof}

\subsection{Proof of Lemma~\ref{l:GammaG}}
\begin{proof}
 $\Lambda_g$ is a symmetric, positive definite and Schur stable matrix \cite{tanwani2021lyapunov}. Since $\tilde{\Lambda}_g = V^T \Lambda_g V$ is similar to $\Lambda_g$, $\tilde{\Lambda}_g$ preserves these properties and thus $I_{N-1} - \tilde{\Lambda}_g \succ 0$ holds, which implies: 
    \begin{gather} \label{eqq:matrix_prop}
    (1- \|V_u \Lambda_g V_u^T\|)(1-\|V_s^T \Lambda_g V_s\|) - \|V_s^T \Lambda_g V_u\|^2 > 0.
    \end{gather} 
    By means of Cauchy's eigenvalue interlacing theorem according to \cite{horn2012matrix}, the relations 
    \begin{align}
    & 0 < \lambda_{min}(V_s^T \Lambda_g V_s) \leq \lambda_{max}(V_s^T \Lambda_g V_s) < 1, \\
    & 0 < \lambda_{min}(V_u^T \Lambda_g V_u) \leq \lambda_{max}(V_u^T \Lambda_g V_u) < 1
    \end{align}
    follow. Combining these relations with \eqref{eqq:matrix_prop} yields $|\gamma_0| < 1$ and $|\gamma_1| < 1 + \gamma_0$ for the coefficients of the characteristic polynomial $P_{\Gamma_g}(s) = s^2 + \gamma_1 s + \gamma_0$ of $\Gamma_g$. Thus, $\Gamma_g$ is a  Schur stable matrix according to the Schur-Cohn criterion \cite{antsaklis2006linear}.   
\end{proof}
\bibliographystyle{IEEEtran}
\bibliography{lit}

\end{document}